\documentclass[10pt, conference, compsocconf]{IEEEtran}

\usepackage[english]{babel}
\usepackage[utf8]{inputenc}
\usepackage{amsmath}
\usepackage{amsthm}
\usepackage{amssymb}
\usepackage{mathtools}

\usepackage{graphicx}
\usepackage{subfig}
\usepackage{floatrow}

\usepackage{float}
\usepackage{algorithm}
\usepackage{algorithmic}

\usepackage{babel}

\usepackage{booktabs}

\newtheorem{theorem}{Theorem}

\newtheorem{lemma}{Lemma}

\author{
\IEEEauthorblockN{Yehezkel S. Resheff}
\IEEEauthorblockA{Intuit \\ 
hezi.resheff@gmail.com}
\and
\IEEEauthorblockN{Moni Shahar}
\IEEEauthorblockA{Intuit \\
monishahar@gmail.com}
}

\begin{document}
\title{A Statistical Approach to Inferring Business Locations Based on Purchase Behavior}




\maketitle

\begin{abstract}
Transaction data obtained by Personal Financial Management (PFM) services from financial institutes such as banks and credit card companies contain a description string from which the merchant identity and an encoded store identifier may be parsed. However, the physical location of the purchase is absent from this description. In this paper we present a method designed to recover this valuable spatial information and map merchant and identifier tuples to physical map locations. We begin by constructing a graph of customer sharing between businesses, and based on  a small set of known “seed” locations we formulate this task as a maximum likelihood problem using a model of customer sharing between nearby businesses. We test our method extensively on real world data and provide statistics on the displacement error in many cities.
\end{abstract}

\section{Introduction}

Personal Financial Management (PFM) services and financial aggregators are software applications that collect and bring together information from multiple sources to provide users with a single stop shop for tracking and managing their personal finances. For individuals with multiple bank accounts, credit cards, and utility bills, seeing the big picture and gaining insights into their financial health can be incredibly valuable. Indeed, services of this sort are used by millions of people in the US alone. 

One of the most important types of information collected and analyzed by PFM services are financial transactions. Bank and credit card transactions are retrieved from financial institutes after users provide the appropriate credentials. These pieces of information essentially sum up to the full financial story pertaining to an individual (even in the case of small transactions where cash still dominates \cite{wakamori2017shoppers}, ATM withdrawals are still recorded, and they tell part of the story).

Across the plethora of financial institutes in the US, the information consistently retrieved by the service is the date, a dollar amount, and a varying length string describing the transaction. These strings are semi human-readable, and include information such as the merchant, time-stamps, and other pertinent information including a (typically numeric) store identifier when a purchase is made in a chain store. 

A piece of information that is not directly present in the transaction data, and that is of extreme importance to PFM services is precise location information. When a purchase is made at a chain store, the identifier accessible in the transaction description string is not directly mapped to a physical location via publicly accessible directories and data sources. A database of all businesses in the US is commercially available via several providers. However, the list of stores in a chain or franchise will not contain these arbitrary internal identifiers. The task we tackle in this paper is to infer store locations based on purchasing patterns across users.  

The need for location information arises in many aspects of the activity of a PFM. First and foremost as an enabler of personalization, and recommendation of more relevant content. This information is also useful in various tasks such as fraud detection, advertising, and user profiling to name a few. Moreover, the location of the businesses and individuals, together with the purchase data can be used for higher level economical analysis both of stores and of populations. 

\section{Data collection and processing}

\begin{figure*}[h]
\label{fig:map}
\centering{}
\includegraphics[width=16cm, trim=0cm 0cm 0 0, clip]{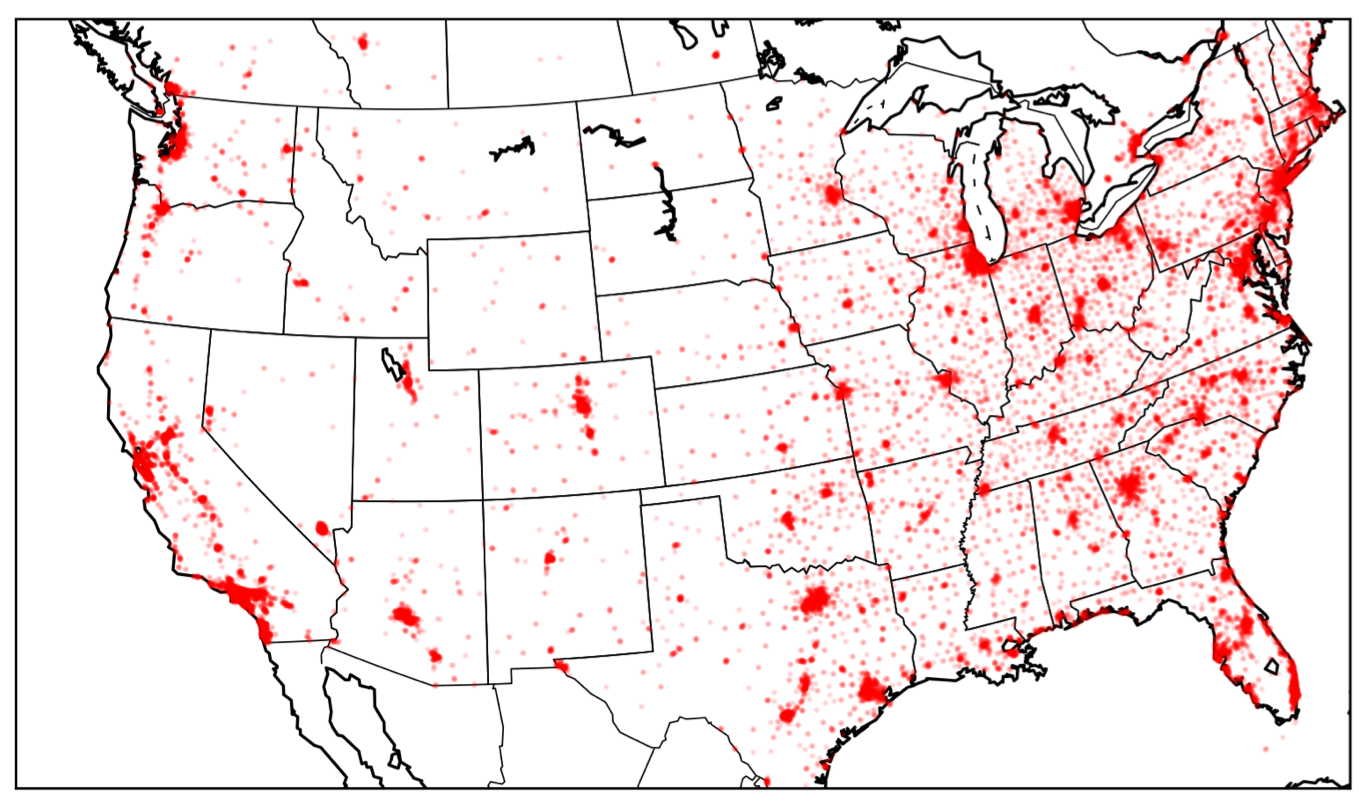}
\caption{The spatial distribution of approximately $40,000$ shops is the dataset used in this paper (Only US stores were used).
}
\end{figure*}

Data used for this work was collected by a large financial data aggregation service. During registration, users provide credentials that allow the aggregator to continuously obtain transaction data from over $25,000$ financial institutions including banks and credit card companies. A record describing a transaction typically contains the date of the purchase, a dollar amount, and a \textit{description string} explaining the nature of the transaction. The first step in the process is to extract structured information out of these strings. Namely, we would like obtain the identity of the merchant from which a purchase was made. For chain stores this includes the name and branch identifier (See Table \ref{tbl:pre-processing} for some examples). 

The main difficulty in extracting structured information from the transaction description strings arises from the variability in processing the information undergoes along the way. The identity of the merchant, as well as the financial institutions processing the information en route, all affect the structure and the information available in the final string obtained by the service. As a result, a number of heuristic and machine learning methods are used to obtain structured information. The extracted information includes fields describing where and when the transaction took place such as exact time-stamps, location information, merchant name, and branch identifier. 

The information we use in this work are the set of purchases for individual users, each characterized only by the the merchant and branch-ID. These branch identifiers uniquely represent physical real-world locations, but the mapping to the real world is unknown. The main task we tackle in this paper is to utilize the full set of data to infer the real world locations of these arbitrary identifiers, based on the key insight that individual stores that share a large percent of their customer base are likely to be near each other. 

Overall, available data contains over 15,000,000,000 transactions per year, arriving from over 10,000,000 users. This represents several percent of all private transactions in the US. In our experiments (Section \ref{sec:experimental}), we use slices of this data pertaining to specific chains in well-defined geographical areas. All experiments were conducted with data from year 2017, and the first quarter of 2018. 

A ground truth location dataset is constructed by scraping websites of the relevant store chains and obtaining both store identifiers and address information. This approach is limited for two reasons. First, the process of building the mechanism to scrape each website is time consuming since the tool has to be adapted to each new store chain. More importantly though, only a small fraction of chains presents the stores with both their addresses and the internal identifiers we see in the transaction data. The ground truth dataset is used in this work to test the methods (Section \ref{sec:experimental}). It is however also an important component in the general inference system for the millions of store locations in the US, since it serves as as seed with which we can start inferring new locations. The methods we use to do this are presented in the next section.

\begin{table*}
\caption{Pre-processing includes the extraction of the merchant name, and store identifier when possible. For some transactions, either one or both of these may not be available; others include additional information such as city or state. \textbf{Our ultimate goal in this paper is to fill in the rightmost column of this table} -- the exact coordinates of the shop. Note this table is an illustration only, and doesn't contain real customer data.}
\label{tbl:pre-processing}
\begin{center}

\begin{tabular}{|l|l|r|r|r|}
\hline 
description string & merchant & store ID & city & exact location         \tabularnewline \hline \hline 
Starbucks Store 06607 & Starbucks & 06607 & ? & ?                  \tabularnewline \hline 
Starbucks Coffee  & Starbucks & ? & ? & ?                         \tabularnewline \hline 
MCDONALD'S M6793  & McDonald's & 6793 & ? & ?                     \tabularnewline \hline 
PIZZA HUT 030579  & Pizza Hut & 030579 & ? & ?                     \tabularnewline \hline 
SHELL OIL 5908 SAN DIEGO CA & Shell & 5908 & San Diego & ? \tabularnewline \hline 
The Who Knows where shoe store     & ? & ? & ? & ?                \tabularnewline \hline 
\end{tabular}

\end{center}
\end{table*}

\section{Methods}
\label{sec:methdos}

\begin{figure}[h]
\centering{}
\includegraphics[width=9cm, trim=2cm 0 0 0, clip]{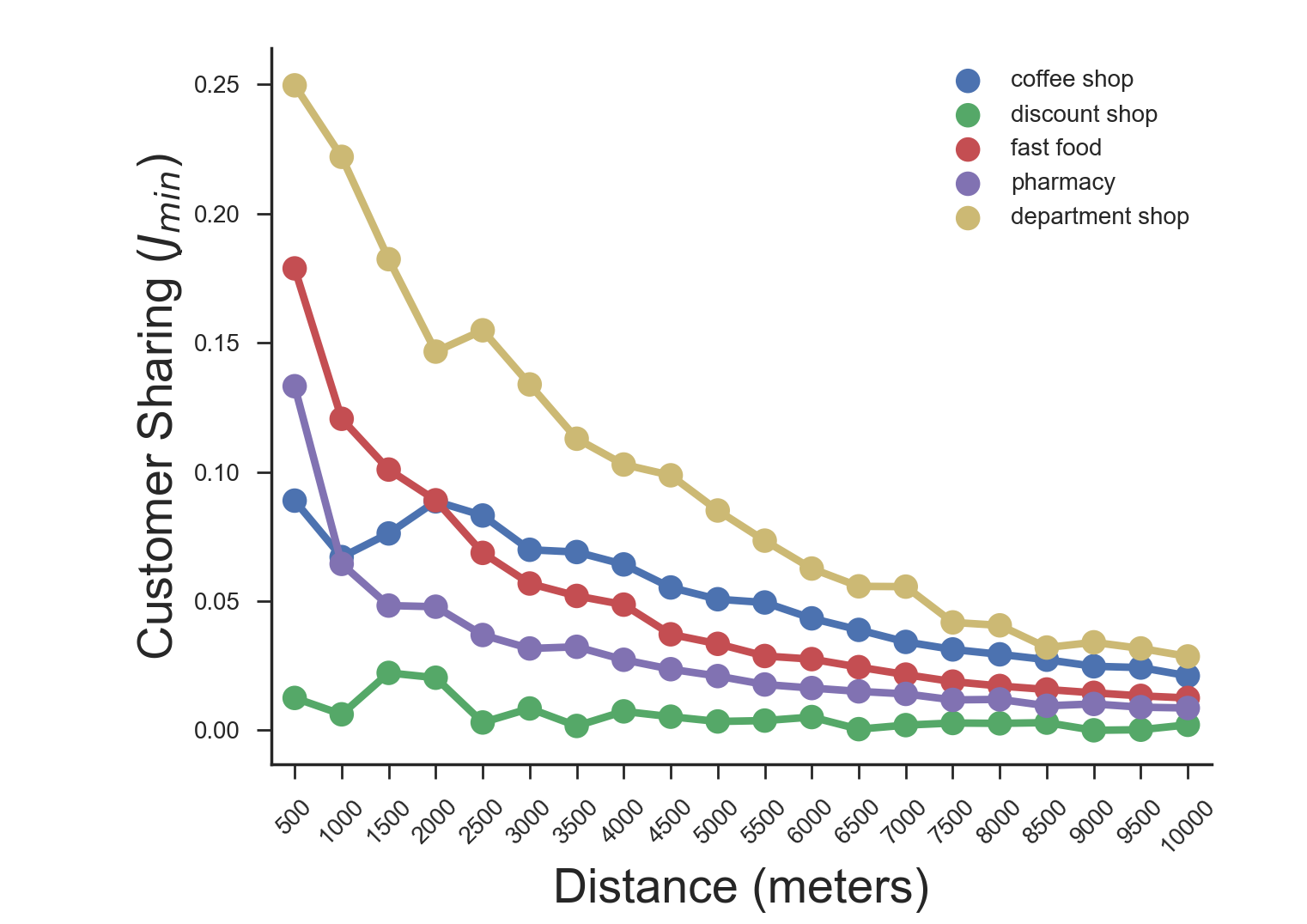}
\caption{Average Customer sharing as a function of distance from a focal shop for $5$ large chains. In all cases the focal shop is from the coffee shop chain. Customer sharing is measured as the percent of customers from the smaller of the two shops (see formula \ref{eq:j-min}).
}
\label{fig:mfd}
\end{figure}

\begin{figure}
\centering{}
\includegraphics[width=9cm, trim=2.5cm 0cm 1cm 2cm, clip]{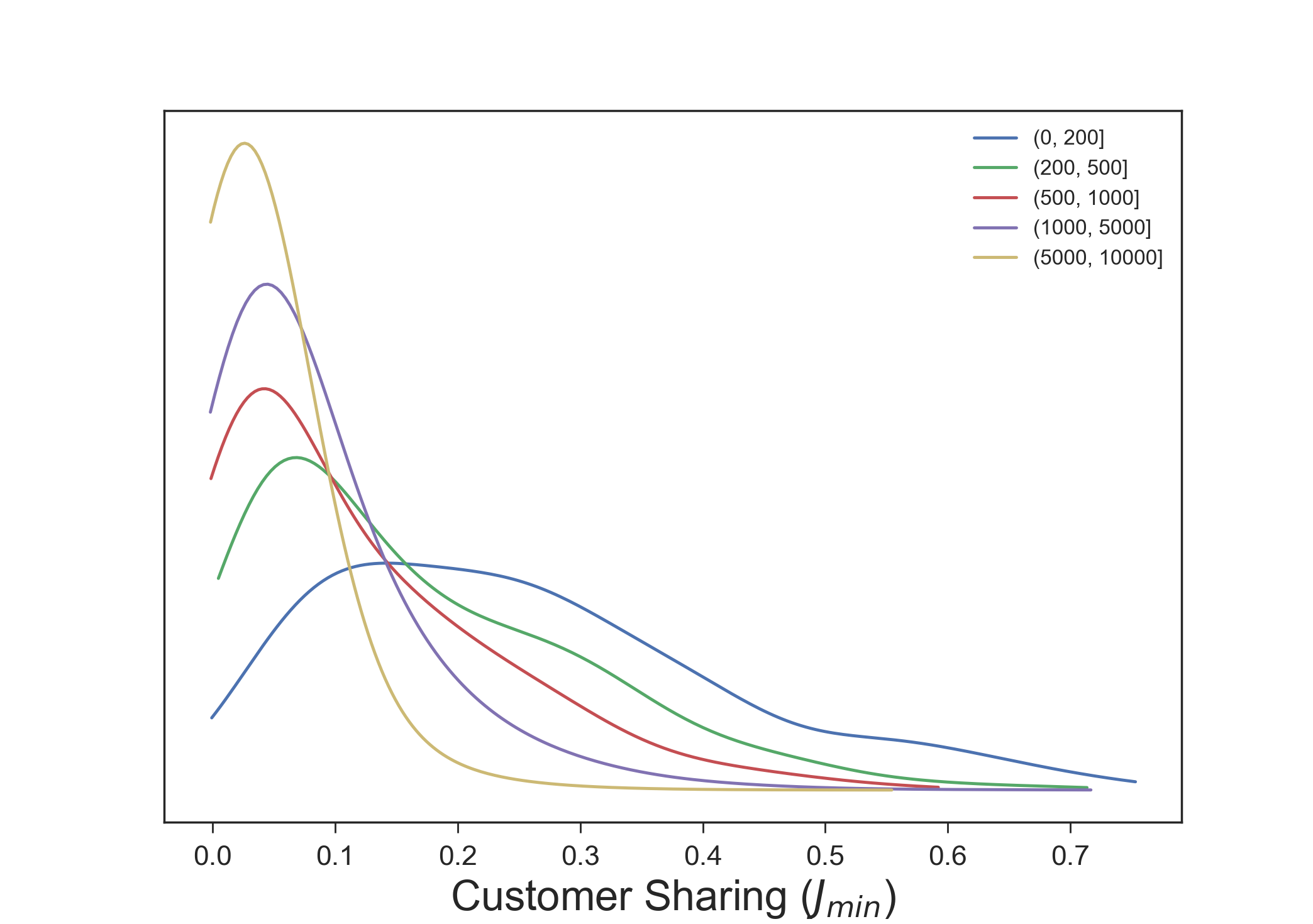}
\caption{$p(m|r)$: Kernel Density Estimates (KDE) for radius binned in values 0-200, 200-500, 500-1000, 1000-5000, 5000-10000. All distance values here and throughout the paper are in meters. }
\label{fig:kde}
\end{figure}

In this section we develop a statistical method used to infer locations based on user transactions. The method assumes a subset of known location, and a store-customer relationship matrix, indicating which users are customers of which stores. The subset of known locations arises in our data via chains with public facing store-IDs that are maintained also in the transaction data, and thus allow us to map the signature to physical locations. The store-customer relationship matrix is a Boolean indicator matrix, where the $(i,j)-th$ position determines whether user $i$ frequents store~$j$. 

Two key insights allow us to develop the proposed model. The first is the assumption of locality of consumer behavior, and the second is the real-world metric of locality. By locality of consumer behavior we mean that, in the probabilistic sense, a customer of a given shop $x$ is more likely to be a customer of a shop $y$ in the same neighborhood, than of a shop $z$ which is far away. This intuitive insight is backed by the data, as seen in Figure \ref{fig:mfd}. 

The second insight is that Euclidean distance is not a good determinant of locality. In a dense urban environment the distance between shops is small, and even at modest distances we don't expect to see much overlap in customer base. In a more rural area on the other hand, the distance between shops a single user frequents can easily span distances on the order of the radius of a a city. 

As a result, we must either consider separately areas with different density (either by population or business place density), or otherwise enter this source of variability into our model. In the current work we maintain a rather homogeneous density by inferring locations for stores one city at at time, but our framework can easily accomodate extra paramters and given enough data can even learn their effect.  

\subsection{Maximum Likelihood Inference of Unknown Store Locations}
\label{sub-sec:max-like}

We consider a dataset ordered as a tuple:
\[
X = (p, q, M)  
\]

\noindent where $p = \{p_i\}_{i=1}^{P}$ is a set of known stores and locations, $q = \{q_i\}_{i=1}^{Q}$ index the unknown locations, and $M \in \mathbb{R}^{P \times Q}$ gives the customer sharing matrix between stores in the sets of known and unknown locations. This matrix is a representation of the bipartite graph where stores are nodes and edges exist between each store with a known location $p_i$ and store with an unknown location $q_j$ -- the element $m_{ij}$ in the $(i,j)$-th position of the matrix $M$ -- giving the customer sharing index between the two stores. As a measure of customer sharing it would be natural to use the Jaccard index: 

\[
	 J(i,j) = \frac{|c(p_i) \cap c(q_j)|}{|c(p_i) \cup c(q_j)|}
\]

\noindent where $c(\cdot)$ denotes the set of customers of a store. The Jaccard index notion of customer sharing of stores does however have some pitfalls. For instance, when comparing a small shop to a large one, even if all the customers of the small shop also frequent the large shop, the computed similarity will still be small. For this reason we propose an alternative minimum-normalized version:

\begin{equation}
\label{eq:j-min}
	m_{i,j} \equiv J_{min}(i,j) = \frac{|c(p_i) \cap c(q_j)|}{min(|c(p_i)|,|c(q_j)|)}
\end{equation}

\noindent which better captures the desired notion of customer sharing when stores are very different in size of customer base. An important component of the model will be the conditional distribution of customer sharing (denoted $m_{ij}$), given the locations of a known shop $p_i$ and an unknown shop $q_j$. We will assume that this distribution depends only on the distance between the two shops, namely:

\begin{equation}
\label{eq:pmgivenr}
	p(m_{ij}|p_i,q_j) = p(m_{ij}|dist(p_i, q_j)) = p(m_{ij}|r_{i,j})
\end{equation}

\noindent here and elsewhere, we use $r_{i,j}$ to denote the distance between the known location $p_i$, and the unknown location $q_j$. in practice, we estimate $p(m_{ij}|r_{i,j})$ empirically based on the set of known locations and customer sharing between them. The joint likelihood of the dataset is then:

\begin{equation}
\label{eq:log-like-0}
	Pr(p,q,M) = \Pr(p) \Pr(q) \Pr(M|p,q) \propto \prod_{i=1}^{P} \prod_{j=1}^{Q} p(m_{i,j}|r_{i,j})
\end{equation}

\noindent where we assume a flat and i.i.d. distribution of locations of both known and unknown entities. Recall that our aim is to resolve the unknown locations $q$:

\begin{equation}
\label{eq:log-like-1}
	q^* = \underset{q}{\max} \log \Pr(p,q,M) = \underset{q}{\min} -\sum_{i=1}^{P} \sum_{j=1}^{Q} \log p(m_{i,j}|r_{i,j})
\end{equation}

At this point the problem statement indicates all $P$ stores are important for the determination of each of the $Q$ unknown locations. However, the problem is further simplified by revisiting the locality property of purchase behavior. Overall, a shop of unknown location $q_j$ will have significant customer overlap with very few of the known locations $p_i$, meaning $M$ is essentially a sparse matrix. We will thus only consider those locations where $m_{i,j} \geq \theta$ for some threshold $\theta$. Customer sharing values below this threshold are treated as zeros. The correct value for this parameter is determined experimentally. The need for such a threshold also arises statistically, since rare co-occurrence where a customer happened to make a purchase at an unusual location cause small sample size effects that make inference very noisy, thus the introduction of the threshold to the customer sharing values adds robustness to our model. 

\noindent Finally, taking the threshold on customer sharing into account the problem becomes:

\begin{equation}
\label{eq:log-like-alpha}
	q^* = \underset{q}{\min} -\sum_{j=1}^{Q} \sum_{i: m_{ij} \geq \theta} \log p(m_{i,j}|r_{i,j})
\end{equation}

\noindent We now briefly discuss some properties of formulation (\ref{eq:log-like-alpha}):

\begin{enumerate}
\item Separability: since $r_{ij} = dist(p_i, q_j)$, and $p_i$ is a store with known location, the problem stated in (\ref{eq:log-like-alpha}) is completely separable in the unknown stores $q_j$. This means we may resolve each unknown point independently of all others. This property has consequences to the scalability of the approach to big datasets by enabling a trivial parallelization of the store location inference. In short, the process begins with a global step of estimating the underlying conditional customer sharing probabilities $p(m|r)$, then the unknown locations are segmented into batches and their locations are inferred. 

The drawback of this parallel approach is that information is not shared between \textit{unknown} locations. Ideally, we would want each new unknown location we infer to benefit from the locations of the already-inferred places. Note that in this case the graph is no longer bipartite, however this doesn't change the formulation. A trade-off between the need for scalability and local information sharing of this sort is achieved by segmenting the data by areas, and processing each area (such as city or county) separately. This is the approach we take (see Section \ref{sec:experimental}). It is however possible to treat the already inferred locations as a less trusted information (relative to locations which are known as a fact) for the inference of the current location, by giving a lower weight to such terms in the sum. For the sake of simplicity we didn't use this information in the procedure we adopted, and the experimental results presented in this paper.   

\item In the single neighbor limit where for an unknown point $q_j$ there is a single point $p_i$ for which $m_{ij} \geq \theta$, the problem formula for resolving that point becomes: 
\[
q^* = \underset{q}{\min} -\log p(m_{i,j}|r_{i,j})
\]
\noindent which is the circle around the point $p_i$ of radius $r$ determined by an argmin operation over $f(r) = p(m_{i,j}|r)$, with $r = dist(p_i, q_j)$. Intuitively, since each known locations active in the problem gives a single radius constraint of this sort, in the general case we should need at least $3$ locations (in general position), in order to be able to uniquely determine the unknown location. In the single known point case however, for a large enough $\theta$, $p(m_{i,j}|r)$ is maximal for $r=0$, and thus the solution collapses to the single known point $p_i$ (this property of $p$ is further discussed bellow, and can be seen empirically in Figure \ref{fig:kde}). 

\item In many cases it is reasonable to force the inference problem for an unknown point to be based only on a few known points in its vicinity. Technically, this may be achieved by selecting a large value of $\theta$, i.e. using only shops with a large shared customer group with the unknown location under consideration (furthermore, the value of the threshold could in theory be chosen dynamically per location in order for the problem to contain a number of known locations within a pre-determined boundary). In this region, the solution to \ref{eq:log-like-alpha} is in the convex hull of the set of known locations used. To see this, we note that in the high threshold region of the problem, values $m>\theta$ of customer sharing are monotonically more likely as the distance between the shops decreases (this can be seen empirically in Figure \ref{fig:kde}). We formalize this notion:

\begin{theorem}
\label{thm:1}
Let $\{p_i\}_{i=1}^{P}$ be a set of known locations, and $q$ an arbitrary unknown location. Assume that $p(m|r)$ is monotonically decreasing in $r$ s.t. $\forall m>\theta, r_1<r_2: p(m, r_1) \geq p(m, r_2)$, then the solution to the problem $\underset{q}{\min} f(q)$ where $f(q) = -\sum_i \log p(m_i, dist(p_i, q))$ is in the convex hull of  $\{p_i\}_{i=1}^{P}$.   
\end{theorem}

\begin{proof}
Suppose, on the contrary, that the solution $q^*$ is not in the convex hull of $\{p_i\}_{i=1}^{P}$. By the following lemma (Lemma \ref{lemma:1}.) the projection of $q^*$ onto the convex hull is closer than $q^*$ is to \textbf{each} of the points $\{p_i\}_{i=1}^{P}$, and thus by the monotonicity of $p(m|r)$ in $r$ we get that each element in the sum $ \sum_i -\log p(m_i, dist(p_i, q))$ is decreased, and therefore $f(q)$ is decreased, in contradiction to $q^*$ being the optimum.

\begin{lemma}
\label{lemma:1}
Let $S \subset R^n$ be a compact convex set. $\forall p \in S, p_0 \notin S: ||p-p_0|| > ||p-p_1||$, where $p_1$ is the projection of $p_0$ onto $S$. 
\end{lemma}

\begin{proof}
Let $p_1$ be the projection of $p_0$ onto $S$, i.e., the closest point in $S$ to $p_0$. Note that such a point exists due to the compactness of $S$. From compactness it is easy to see that $p_1$ is on the boundary of $S$, for otherwise there exists a ball of radius $\epsilon$ such the $B_\epsilon(p_1) \subseteq S$ on which there is a point closer to $p_0$. From the Hahn-Banach theorem there exist a supporting hyper plane $P$ such that $P\cdot p_0 \ge c = P \cdot p_1 \ge P \cdot p \quad \forall p \in S$ . Consider the triangle formed by $p_0,p_1,p$. The angle $\angle p_0p_1p$ is the largest angle (since a supporting hyperplane of $S$ at $p_1$ separates the set), and therefore the edge  $pp_0$ which is opposite to it is larger than the proximal edge $pp_1$. 
\end{proof}

\end{proof}

\noindent This concludes the proof of Theorem \ref{thm:1}. As a consequence of this theorem, we must be mindful when selecting values of the threshold $\theta$ that imply a solution in the convex hull when inferring for stores in areas where the extent of coverage of the known locations is insufficient for us to assume that the unknown location is in their convex hull. On the other hand, we will use larger values of $\theta$ and the convex hull property to help in resolving locations in other cases, this is further discussed in the following items.

\item A key quantity we use throughout this discussion is the conditional distribution of customer sharing given store distance $p(m_{i,j}|r_{i,j})$. Up until now we have assumed it to be known, however in practice this quantity must be estimated from the data. More specifically, we look at the customer sharing and distances between pairs of \textit{known} locations. Both distance ($r$) and customer sharing ($m$) are binned, and the appropriately normalized frequency table is computed once and stored for the computation of each inferred location.

\item Finally, when it comes to resolving an unknown location, the natural solution for maximum likelihood problems of this sort is to use gradient based methods; in this case:

\begin{equation}
\begin{split}
    \frac{\partial}{\partial q} &\sum_i \log p(m_i|r_i) = \\
- & \sum_i \frac{1}{p(m_i|r_i)} \frac{\partial p(m_i|r)}{\partial r} \frac{\partial r}{\partial q} \Bigg\rvert_{r=r_i} 
= \\ 
& \sum_i \frac{p'(m_i|r_i)}{p(m_i|r_i)} \frac{p_i - q}{r_i}
\end{split}
\end{equation}

\noindent Leading to the gradient descent algorithm (Algorithm \ref{alg:1}). It is noteworthy that the computation of the partial derivative $\frac{\partial}{\partial r}p(m|r)$ is only possible numerically since we do not assume any functional form of this conditional distribution. One could instead select a parametric family for this distribution and proceed with an analytic partial derivative. In the analytic case, the binned estimation of the conditional distribution described above would be replaced by estimating the parameters. As a consequence, for numeric stability we must estimate $p(m|r)$ in relatively narrow bins of $r$. This leads to a small amount of data in each bin and hence again to stability issues. 

These problems make the gradient descent solution less appealing for our specific problem. Luckily, by Theorem \ref{thm:1} we have that for large enough values of $\theta$ the solution is in the convex hull of the customer-sharing neighbors of the unknown shop. In this region we are looking for a solution in a relatively contained area in $2D$, and can apply an exhaustive (or in practice grid) search. We empirically find that a threshold of $\theta = .15$ is sufficient for the convex hull property to hold (see also Figure \ref{fig:kde}). 

\begin{algorithm}[tbh]
	\begin{algorithmic}[1] 
    	\STATE{$q \leftarrow q_0$}
        \WHILE{not converged}  
        	\STATE{$\forall i: r_i \leftarrow dist(p_i, q)$}
            \STATE{$q \leftarrow q - \alpha \sum_i \frac{p'(m_i|r_i)}{p(m_i|r_i)} \frac{p_i - q}{r_i}$}
        \ENDWHILE
    \end{algorithmic}
    
    \protect\caption{Gradient Descent algorithm for maximum likelihood unknown store localization.}
	\label{alg:1}
\end{algorithm}

\end{enumerate}

\begin{figure}[h]
\label{fig:graph}
\centering{}
\includegraphics[width=10cm, trim=0cm 0 0 0, clip]{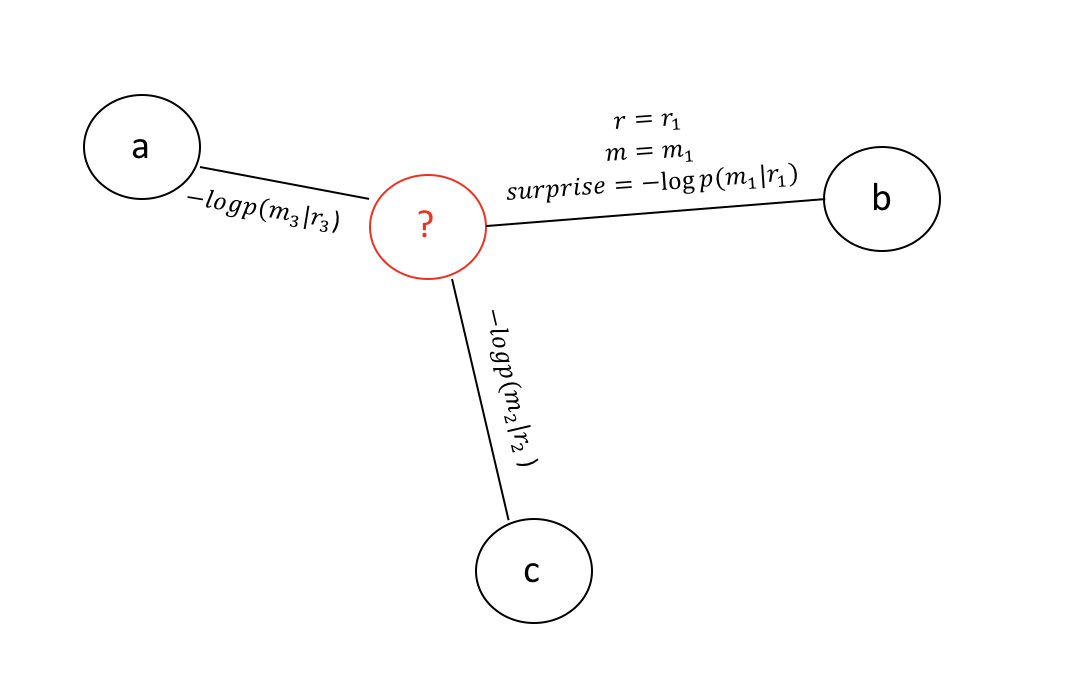}
\caption{The geometric interpretation of the maximum likelihood formulation (Equation~\ref{eq:log-like-alpha}). Our aim is to find a location assignment in the convex hull of the neighboring stores (in the customer sharing space), such that the sum of the surprise of the observed store sharing index ($m_i$) given the physical world distances between the assignment and the known locations ($r_i$) will be minimized. }
\end{figure}

\section{Experimental Evaluation}
\label{sec:experimental}

In this section we present a thorough evaluation of the proposed method. Ground truth locations were obtained for a subset of approximately $40,000$ shops (a tiny portion of the many hundreds of thousand merchants in the US), including a large coffee shop chain, a fast food franchise, a pharmacy, a department store, and a discount super-market. The store to store customer sharing graph was constructed for the entire list of stores in this dataset, and thresholded at $\theta=0.15$, this value was decided based on Figure \ref{fig:mfd} as the minimal value where customer sharing is monotonically more likely as distance between shops decreases, a property important in our algorithm for resolving locations (see Section \ref{sub-sec:max-like} above).  

We test the proposed method on stores from two lists of cities. The first experiment is conducted on a list that is representative of large cities in the US. The second consists of large to medium cities in the state of California. Together, these experiments test the method on cities representing the core user base of the financial data aggregation service the data is obtained from. Table \ref{tbl:cities-descriptive} lists the cities used in this section. The number of stores per city in the dataset varies substantially between $509$ (Chicago) and $64$ (Bakersfield, CA). 

Another variable that is expected to have a strong effect on the ability to recover locations is the store density in each area. This value is estimate as average of either the mean or the median distance to the nearest store in the data, and is calculated per city. The median distance ranges between $517$ meters (Fresno) to $0$ meters (Chicago). Note That value of $0$ distance is obtained when two shops are in the same shopping center or mall, and thus share the same address, but the real world distance can be up to several hundred meters (for our purpose, placing everything in a mall at the same location makes sense, because the effect of distance on customer sharing within this space is expected to be minimal).

The average value of customer sharing (Final column in Table \ref{tbl:cities-descriptive}) with the nearest shop seems not to directly relate to the store density value. We interpret this as an indication that distances have a different meaning to customers depending on the environment. The tendency to walk or travel longer distances in a specific city should ideally be taken into account. We further discuss this issue with respect to future work in the Conclusions section bellow. 

In each of the following experiments, the ability of the proposed method to infer store location was tested and compared to two baseline options. The \textit{NN-1} baseline consists of the nearest neighbor in customer sharing space. Namely, for each test store the location of the known store with the highest value in the customer sharing index is adopted as the inferred location. Likewise, the \textit{NN-3} method uses the center of mass of the three stores with the largest extent of customer sharing with the test store, and adopts that as the inferred location for the unknown store.

The tests are preformed in a leave-one-out scheme. One at a time, each of the stores in the dataset is treated as an unknown location, and the location is to be inferred by each of the methods based on all other stores. The misplacement error is then computed as the distance between the inferred location and the actual location of the store. This testing scheme is meant to approximate the error expected when utilizing these methods to find the locations of stores with currently unknown locations, based on the full dataset of known locations.

The California cities evaluation included Los Angeles, San Diego, San Francisco, San Jose, Sacramento, Fresno, and Bakersfield. Firstly, for all cities in the list both the baseline \textit{NN-1} and the proposed method (\textit{MaxLike} column in table \ref{tbl:results-cali}) achieve the desired neighborhood precision. This serves as an indication that the general method of inferring location based on shred customers is feasible. Secondly, in all cases the proposed maximum likelihood based method outperforms both baseline methods, often by a substantial margin.

The purpose of the \textit{NN-3} baseline is to determine whether the proposed maximum likelihood method succeeds mostly due to the integration of information from several nearby locations. However, here and in all other experiments we preformed, the \textit{NN-3} method is strongly dominated by \textit{NN-1}, meaning that if we are to simply use weighted average location, then adding extra locations is of no use. This result persisted also when going beyond $3$ neighbors, further indicating the usefulness of our statistical model.

\begin{table*}
\caption{Descriptive statistics for the data used in the experiments. n - number of stores per city in the dataset, dist(median) - the median distance in meters to the nearest known location, dist (mean) - the mean distance in meters to the nearest known location, $J_{min}$ - the average customer sharing index with the nearest location.}
\label{tbl:cities-descriptive}
\begin{center}

\begin{tabular}{|l | c | c | c | c |}
\hline
City          & n & distance (median) & distance (mean) & customer sharing ($J_{min}$)  \tabularnewline \hline \hline 
New York      & 339 & 115  & 297 & 0.092      \tabularnewline \hline   
Chicago       & 509 & 0    & 531  & 0.124      \tabularnewline \hline       
Houston       & 449 & 101  & 1567 & 0.161  \tabularnewline \hline       
Philadelphia  & 120 & 172  & 514  & 0.134  \tabularnewline \hline       
Phoenix       & 215 & 70   & 520  & 0.153  \tabularnewline \hline       
Los Angeles   & 219 & 254 & 500 & 0.098  \tabularnewline \hline   
San Diego     & 178 & 192 & 538 & 0.112   \tabularnewline \hline       
San Francisco \: & \: 174 \: &  96 & 190 & 0.138   \tabularnewline \hline       
San Jose      & 135 & 179 & 355 & 0.118    \tabularnewline \hline       
Sacramento    & 120 & 208 & 486 & 0.162    \tabularnewline \hline   
Fresno        & 70  & 517 & 737 & 0.144   \tabularnewline \hline
Bakersfield   & 64  & 281 & 635 & 0.186    \tabularnewline \hline
\end{tabular}

\end{center}
\end{table*}

\begin{table}[ht!]
\caption{Median displacement error for California cities.}
\label{tbl:results-cali}
\begin{center}

\begin{tabular}{|l|r|r|c|}
\hline
City            & NN-1 & NN-3 & MaxLike  \tabularnewline \hline \hline 
Los Angeles     & 1034 & 1743 & 741      \tabularnewline \hline   
San Diego       & 1509 & 2715 & 1268     \tabularnewline \hline       
San Francisco   & 269  & 683  & 255      \tabularnewline \hline       
San Jose        & 1473 & 1845 & 1283     \tabularnewline \hline       
Sacramento      & 1413 & 2110 & 1162     \tabularnewline \hline   
Fresno          & 2198 & 3331 & 1976     \tabularnewline \hline
Bakersfield     & 1673 & 2556 & 1265     \tabularnewline \hline

\end{tabular}

\end{center}
\end{table}

The median displacement error in inferring city location in California (Table \ref{tbl:results-cali}) are highly variably as we would expect from the varying sample size and density (Table \ref{tbl:cities-descriptive}). In San Francisco -- a dense city with a large sample size -- we obtain a median displacement of $255$ meters. This goes up to nearly $2$ Kilometers in Fresno, where the sample size is only $70$, and the stores are sparsely distributed with a mean distance of $737$ to the nearest store in this dataset. 

These two parameters (number and density of known locations in a region) emerge as important determinant of success both for the proposed method and the baseline alternative. Essentially, since our approach is to infer locations based on neighbors, we will inherently be limited by the quality of our known neighbor data. This has practical consequences in directing the type of additional information we should collect in order to improve results where they are most important to us. 

For the most part, for various reasons we are interested in locating shops to within the typical distance between shops of the same chain. The per-chain results are shown for the maximum likelihood method in Table \ref{tbl:results-cities-by-chain}. Results for the pharmacy chain are excellent, mostly under $1000$ meters. Sine pharmacy shops of the same chain do not tend to appear in close proximity to each other, these numbers are satisfactory.

For the coffee shop franchise on the other hand we see mixed results. In some cities we are able to infer their locations with median accuracy in the $200-400$ meter range. In other cities we are substantially above $1000$ meters. For the urban densely of popular coffee shops this may well not be sufficient. This directs us to the type of data we must collect, namely additional shops in these specific areas where coverage is currently insufficient for our needs.

In the large city experiment (Table \ref{tbl:results-cities}) results are very similar in essence. Overall the proposed method outperforms the baseline in almost all cases. In a single case (Philadelphia) the $NN-1$ baseline ties with the maximum likelihood approach. Here again results are tied to a large extent to the number and density of the known locations in the city in our dataset. The great accuracy of $173$ meters in NYC for instance is tied to the large sample size and high density in the dataset (See Table \ref{tbl:cities-descriptive}), whereas for Houston we have a large sample size, but low store density, and a median displacement error of $1339$ meters.

\begin{table}
\caption{Median displacement error for large US cities.}
\label{tbl:results-cities}
\begin{center}

\begin{tabular}{|l|r|r|c|}
\hline
City          & NN-1 & NN-3 & MaxLike  \tabularnewline \hline \hline 
New York      & 229  & 297  & 173      \tabularnewline \hline   
Chicago       & 503  & 1161 & 442      \tabularnewline \hline       
Houston       & 1504 & 2510 & 1339     \tabularnewline \hline       
Philadelphia  & 1033 & 1977 & 1033     \tabularnewline \hline       
Phoenix       & 1686 & 2435 & 1069     \tabularnewline \hline       

\end{tabular}
\end{center}
\end{table}

\begin{table*}
\caption{Median displacement error for California cities by chain type. The department shop chain was available in the dataset for some of the cities only.}
\label{tbl:results-cities-by-chain}
\begin{center}

\begin{tabular}{|l|r|r|r|r|}
\hline
city/chain &  coffee shop &  fast food &  pharmacy &  department shop    \tabularnewline \hline \hline
Los Angeles   &        372 &       1253 &        268 &       -  \tabularnewline \hline
San Diego     &       1182 &       1975 &        960 &      247 \tabularnewline \hline
San Francisco &        196 &        917 &        291 &       -  \tabularnewline \hline
San Jose      &       1121 &       1954 &       1335 &     1117 \tabularnewline \hline
Sacramento    &       1101 &       1929 &       1064 &     1555 \tabularnewline \hline
Fresno        &       1821 &       3170 &       1609 &       -  \tabularnewline \hline
Bakersfield   &       1238 &       1335 &        690 &       -  \tabularnewline \hline
\end{tabular}
\end{center}
\end{table*}

\section{Related Work}

{\bf The Weighted Graph Matching Problem (WGMP)} A weighted graph is an ordered pair $(V,W)$ where $V$ is a set of nodes and $W$ is a function $W : V \times V \rightarrow \cal{R}_+$. If the function $W$ is symmetric the graph is undirected, whereas for an arbitrary $W$ the graph is directed. The Weighted Graph Matching Problem (WGMP) \cite{bunke2000graph} is the following: given two weighted graphs $(V_1,W_1)$ and $(V_2,W_2)$  where $|V_1|=|V_2|=N$, find a one-to-one correspondence $\Phi : V_1 \rightarrow V_2$ that minimizes: 

$$J(\Phi) = \sum_{i=1}^N \sum_{j=1}^N (W_1(v_i,v_j) - W_2(\phi(v_i), \phi(v_j)))^2$$

A closely related problem is the Graph Isomorphism Problem (GIP) where the goal is to decide if for a two given graphs $G_i=(V_i,E_i)$, there exists a one-to-one correspondence between their verticals $\phi$ such that $(v_i,v_j) \in E_1 \iff (\phi(v_i),\phi(v_j)) \in E_2$. GIP can easily be formulated as a version of WGMP for which $W_1,W_2 \in \{1,\infty\}$ where $(v_i,v_j) \in E \iff w(v_i,v_j) = 1$. 

Both GIP and WGIP (as well as other versions of them) were extensively studied. These problems appear naturally in applications of pattern recognition, computer vision, neuroscience and others (see for example \cite{berg2005shape}, \cite{zhou2012factorized}, \cite{conte2004thirty}). From theoretical aspect the question whether GIP is in $P$ or $NPC$ is one of the oldest open questions in complexity theory ~\cite{garey2002computers}.  
WGMP is $NP-$Hard, and despite many attempts (see for example \cite{fiori2015spectral},\cite{lyzinski2016graph}), no constant ratio approximation algorithm is known for it.  

{\bf Inferred Location Problems} One of the motivating use-cases for the inference of \textit{store} location in the current work, is to extend this layer of spatial information to \textit{users}, by  associating them with the locations at which their transactions take place. The resulting representation is a heatmap describing a small number of areas which together represent the regular habitat of a person. Previous work has already used purchase history information to predict demographic characteristics of consumers \cite{resheff2017fusing,wang2016your}, and this work may be seen as a further extension of this to geo-spatial information. 

Another field where locations are being inferred from relation graphs combined with a subset of known locations is in digital archeology. An especially inspiring case is that of the Old Assyrian long-distance trade routes. These routes spanned vast distances, between numerous cities and outposts, many of which have since been covered by dust \cite{kool2012old}. Four thousand years later, a combination of Archaeological and Economical graph-based analytics techniques were able to uncover lost locations, based on travel logs detailing the time spent along the route \cite{barjamovic2017trade}.  

One should note that the inference problem we tackle in this paper is different from WGMP also in the structure of the data it operates on. Applying WGMP for our problem would require that the input would be two graphs, a graph of locations whose nodes are indexed in an arbitrary way, and which is weighted by the actual distances between the shops, and a graph (derived from the transactions) whose nodes are indexed according to the internal index of the merchant, and weighted according to customer purchases (we formulate this weighting in Section \ref{sec:methdos}). The goal of WGMP is then to find the permutation that maps the arbitrary indices to the merchant indices, and by this get the locations. 

Even if the list of locations of unknown stores were to be given, the hardness of the graph matching alternative would still require an alternative approach to the global matching of these two graphs. In that case, our approach might be seen as a local matching relaxation, where each node is matched to a location in the graph of real-world locations, based on the already-matched locations. 

\section{Conclusions}

For a large financial aggregation service, knowing where transactions took place is a key enabler for a vast number of location based features. This information is not present in the raw data arriving from financial institutes, and thus has to be inferred. We present a statistical inference formulation of the problem, based on a seed of known locations, and the graph of store to store relationships defined by the shared customer base.  

Experiments on many cities show the proposed method is able to locate stores within neighborhood precision or better. Since the inference process is limited by the density of known locations, as the number of those increases we expect to see improvement in the accuracy of place assignment for unknown shops. 

A drawback of the simple approach described in this paper that in reality the meaning of real-world distances varies according to the type of store one is visiting. For instance, one might travel a relatively longer distance to a preferred super-market than to a coffee-shop (evidence of this is indeed seen in Figure \ref{fig:mfd}). As a result, a relatively distant shop of a one type might share a substantial customer base with a local shop, whereas two distant shops of another kind are unlikely to do so.

Although the algorithmic framework is general enough to support this property of customer behavior in the model, that would require estimating separate functions $p(m|r;c_1, c_2)$ for each two classes of stores $(c_1, c_2) \in C \times C$. However, when scaling to the large dataset containing all US stores, and a very large set $C$ of store classes, the $O(|C|^2)$ functions become prohibitively many. 

Furthermore, we would ideally want to share information between some of the \textit{classes} of stores. For instance, super-markets and hardware shops may very well behave in the same way regarding their profile of customer sharing with local coffee-shops. We do not want to assume any clustering of types of shops, but rather learn these relationships from the data. Future work will focus on the extension of the current method to incorporate these and other attributes of shops and not just the distance between them. 

Another avenue for further improvement is the measure of distance between shops. In this paper we use the Euclidean aerial distance between locations \textit{as the crow flies}. A better measure of effective distance should consider road topology, or maybe measure the time to travel between locations instead. In order to incorporating this richer notion of distance we will need to overcome substantial hurdles in the optimization of the resulting optimization problems.  

Future work will also focus on incorporating a notion of trajectory analysis (for instance by building on the trajectory locality notion introduced in \cite{resheff2016online}).By understanding the path of users in the space of merchants they make purchases from, and enforcing speed constraints, we will be able to further improve the store localization method developed in this paper. Our hope is that by adding trajectory and other additional sources of information we will be able to pin-point locations, without giving up the favorable scalability properties of the approach that enable the application to real-world financial big-data.

\bibliographystyle{splncs04}
\bibliography{lib}

\end{document}